%% file: submodular.tex
\documentclass[11pt]{article}

\usepackage{amsmath}
\usepackage{amssymb}
\usepackage{amsfonts}
\usepackage{amsthm}

\usepackage{fullpage}
\usepackage[colorlinks=true,citecolor=blue]{hyperref}
\usepackage{multirow}

\newtheorem{theorem}{Theorem}[section]
\newtheorem{corollary}[theorem]{Corollary}
\newtheorem{lemma}[theorem]{Lemma}
\newtheorem{claim}[theorem]{Claim}

\newtheorem{definition}{Definition}

\def \RR {{\mathbb R}}
\def \F {{\mathcal{F}}}
\def \G {{\mathcal{G}}}

\title{On the Approximation of Submodular Functions}

\author{Nikhil R. Devanur\footnote{Microsoft Research, Redmond ~~~$\{ \text{\tt{nikdev, roysch, mohits}} \}${\tt @microsoft.com}} \and Shaddin Dughmi\footnote{University of Southern California~~~{\tt shaddin@usc.edu}} \and Roy Schwartz$^{*}$ \and Ankit Sharma\footnote{Carnegie Mellon University~~~{\tt ankits@cs.cmu.edu}} \and Mohit Singh$^{*}$}

\begin{document}
\maketitle

\begin{abstract}
Submodular functions are a fundamental object of study in combinatorial optimization, economics, machine learning, etc. and exhibit a rich combinatorial structure.
Many subclasses of submodular functions have also been well studied and these subclasses widely vary in their complexity. Our motivation is to understand the relative complexity of these classes of functions. Towards this, we consider the question of how well can one class of submodular functions be approximated by another (simpler) class of submodular functions. Such approximations naturally allow algorithms designed for the simpler class to be applied to the bigger class of functions. We prove both upper and lower bounds on such approximations.
Our main results are:
\begin{itemize}
\item General submodular functions\footnote{We additionally assume that the submodular function takes value 0 on the null set and the universe.} can be approximated by cut functions of directed graphs to a factor of $n^2/4$, which is tight.
\item General symmetric submodular functions$^{1}$ can be approximated by cut functions of undirected graphs to a factor of $n-1$, which is tight up to a constant.
\item Budgeted additive functions can be approximated by coverage functions to a factor of $e/(e-1)$, which is tight.
\end{itemize}
Here $n$ is the size of the ground set on which the submodular function is defined.
We also observe that prior works imply that monotone submodular functions can be approximated by coverage functions with a factor
between $O(\sqrt{n} \log n)$ and $\Omega(n^{1/3} /\log^2 n) $.

\end{abstract}

\input{introduction.tex}
\input{general_dicut.tex}

\input{symmetric_cut.tex}

\input{budget_coverage.tex}
\input{future_directions.tex}

\bibliographystyle{plain}
\bibliography{submodular-ref}

\appendix
\include{appendix}

\input{monotone_coverage.tex}
\end{document}

%% file: introduction.tex
\section{Introduction}
Submodular optimization problems have been a rich area of research in recent years, motivated by the principle of diminishing marginal returns which is prevalent in real world applications. Such functions are ubiquitous in diverse disciplines, including economics, algorithmic game theory, machine learning, combinatorial optimization and combinatorics. While submodular function can be minimized efficiently, i.e., in polynomial time~\cite{GLS81,Schrijver:2000,IwataFF:2001},
many natural optimization problems over submodular functions are NP-hard, e.g., Max-$k$-Coverage~\cite{NemhauserWF:1978}, Max-Cut and Max-DiCut~\cite{GoemansW:1995}, and Max-Facility-Location~\cite{CornuejolsFN:1977a}. Consequently, many works,
 specifically in the setting of algorithmic game theory~\cite{BuchfuhrerSS10,Dughmi11,DughmiRY11,HoeferK12}, have explored \emph{simpler} subclasses of submodular functions for which the given algorithmic problem can still be well-approximated. Such subclasses of submodular functions have included cut functions of graphs, coverage functions of set systems, budgeted additive functions, matroid rank functions, etc.

Our work is motivated by the question, how complex can a submodular function be? Since this is such a fundamental question, it has been asked in different forms previously. Goemans et al.~\cite{Goemans:2009} consider how many queries to a submodular function are sufficient to infer the value of the function, approximately, at every point in the domain. Balcan and Harvey~\cite{Balcan:2010} focus on the problem of learning submodular functions in a probabilistic model; are few random queries enough to infer the value at almost all points in the domain? Badanidiyuru et al.~\cite{Badanidiyuru:2012} ask whether an approximate \emph{sketch} of a submodular function, or more generally a subadditive function, exists (i.e., can the function be represented in polynomial space)? Seshadri and Vondr{\'a}k~\cite{Seshadhri:2010} consider the testability of submodular functions: how many queries does it take to check whether a function is \emph{close} to being submodular?

We approach this question by noting that not all submodular functions are identically complex and some have been more amenable to optimization than others. Thus, one natural way to characterize the relative complexity of one class of submodular functions w.r.t another,
is to ask how well can a function in the first class be approximated by a function in the second.
Formally, we ask the following question. Given two classes of submodular functions $\F$ and $\G$ (typically $\G\subset \F$), what is the smallest $\theta$, such that for every $f\in \F$, there exists a $g\in \G$ such that $f(S)\leq g(S)\leq \theta \cdot f(S)$ for each $S\subseteq U$? Here class $\G$ would represent the class of submodular functions which are easier to optimize for some problem and class $\F$ would represent a bigger class  which we want to optimize over. We also note that this concept of approximation is not special to submodular functions and can be asked for any two classes of functions. We focus on submodular functions due to their ubiquitous nature in optimization.

Intuitively, this notion of approximation resembles the long and rich line of work that deals with the algorithmic applications of geometric embeddings, in which the goal is to embed hard metric spaces into simpler ones. Some successful examples include embedding general metrics into normed spaces~\cite{Bourgain:1985}, dimension reduction in a Euclidean space~\cite{JohnsonL:1984}  and the probabilistic embedding into ultrametrics~\cite{Bartal:1996,FakcharoenpholRT:2004}.
As in the metric case, a natural byproduct of the above approach is that if there exists an $\alpha$-approximation algorithm for any submodular function in $\G$, then there exists a $(\theta\cdot \alpha)$-approximation algorithm for all functions in $\F$. As an application of our approach, we show how to obtain an algorithm for the online submodular function maximization problem, for general monotone submodular functions~\cite{BuchbinderNRS12}. Previously, results were known for only certain subclasses of submodular functions; see Appendix~\ref{sec:application} for details.





\subsection{Our Results and Techniques}\label{sec:results}
We start by asking how well a general submodular function $f:2^{U}  \rightarrow \mathbb{R}_{+}$ (with the additional property that $f(\phi)=0=f(U)$) can be approximated by a function in the canonical simpler subfamily of non-symmetric submodular functions, cut function of a directed graph. We give matching upper and lower bounds for such an approximation (Theorem~\ref{them:General_DirectedCUT}). Next, we ask the same question for \emph{symmetric} submodular functions vis-a-vis its canonical simpler subfamily, cut functions of undirected graphs. In this case, we provide nearly matching upper and lower bounds (Theorem~\ref{them:Symmetric_UndirectedCUT}). We then move our attention to two subfamilies, budgeted additive functions and coverage functions, both of which,
 as already mentioned in the introduction, have received considerable interest in the algorithmic game theory setting. We show tight upper and lower bounds for approximating budgeted additive functions with coverage functions (Theorem~\ref{thrm:Bugdet_by_Coverage}). These results are summarized in Table~\ref{tab:results}.
 While previous works~\cite{Goemans:2009,Balcan:2010,Badanidiyuru:2012} studied the complexity of submodular functions from different perspectives, they do imply some additional results, both positive and negative, on the approximation of monotone submodular functions by simpler classes of submodular functions (as illustrated in Table~\ref{tab:results} and discussed in detail in Appendix~\ref{sec:monotone-by-budgeted}).

\begin{table}[t]
\centering
\renewcommand{\arraystretch}{1.5}
\begin{tabular}{|c|c|c|c|}
\hline
  \multirow{2}{*}{{\bf Input Class}} & \multirow{2}{*}{{\bf Output Class}} & {\bf Approximation} & {\bf Approximation} \\
  & & {\bf Upper Bound} & {\bf Lower Bound}\\
  \hline
 General Submodular
 & Cut Functions (directed) & $\frac{n^2}{4}$ & $\frac{n^2}{4}$\\
  \hline

  Symmetric Submodular
  & Cut functions (undirected) & ${n-1}$ &$\frac{n}{4}$\\
  \hline

  Budgeted Additive & Coverage & $\frac{e}{e-1}$ & $\frac{e}{e-1}$\\
\hline

Monotone Submodular & Coverage/Budgeted Additive & $O(\sqrt{n}\log n)$~\cite{Goemans:2009} & $\Omega(\frac{n^{1/3}}{\log^2{n}})$~\cite{Balcan:2010,Badanidiyuru:2012}\\

\hline

\hline

\end{tabular}
\vspace{2pt}
\caption{Our results are described in the first three rows. The results in the last row are either implicit in the references or follow as a corollary (Appendix~\ref{sec:monotone-by-budgeted}). When the output class is a cut function of a graph, we assume that the input function $f$ satisfies $f(\emptyset)=f(U)=0$, as every cut function must satisfy this constraint. Here, $n$ denotes the size of the ground set.}
\label{tab:results}
\vspace{-15pt}
\end{table}

Let us now briefly discuss the main techniques that we use to obtain our results. In contrast to previous works~\cite{Goemans:2009,Balcan:2010,Badanidiyuru:2012}, {\em arbitrary} submodular functions, as opposed to {\em monotone} submodular functions, present different challenges. As an illustration, for approximating a submodular function $f$ via a cut function of a graph $G$, consider the case when there is a non-trivial set $\emptyset \neq S\neq U$ for which $f(S)=0$. Then the weight of the cut $(S,\bar{S})$ in $G$ is forced to be  zero. Indeed, {\em all} sets $S$ with $f(S)=0$ must be in a correspondence with cuts in $G$ of value zero.
Thus, given a submodular function $f$, our construction of $G$ optimizes for the minimizers of the submodular function $f$. Surprisingly, this can be shown to give the best possible approximation. For a \emph{symmetric} submodular function $f$, we show that it suffices to use the cut function of a {\em tree} (as opposed to a general undirected graph) utilizing the Gomory-Hu tree representation~\cite{GomoryH61,Queyranne93} of $f$.

For approximating budgeted additive functions by coverage functions, we first give a randomized construction achieving an approximation factor of $e/(e-1)$. We then show that this is the best possible approximation factor  as characterized by a linear program. The  proof of the lower bound of $e/(e-1)$ uses linear programming duality and proceeds by presenting a feasible dual solution to the linear program achieving an objective value of $e/(e-1)$ in the limit. 
We would like to point out that all our results are algorithmic and the claimed approximations can be found in polynomial time given a value oracle for the submodular function.



\subsection{Related Work}

Goemans et al.~\cite{Goemans:2009} considered the problem of how well a given
monotone submodular function $f$ can be approximated when only polynomially many value oracle queries are permitted. They presented an approximation of $O\left( \sqrt{n}\log{n}\right)$, and an improved guarantee of $\sqrt{n+1}$ in the case that $f$ is a matroid rank function. Implicit in this algorithm and relevant to our setting is an approximation of all monotone submodular functions by budgeted additive functions  (Appendix~\ref{sec:monotone-by-budgeted}). The current best lower bound for the problem studied by~\cite{Goemans:2009} is given by Svitkina and Fleischer~\cite{SvitkinaF:2011} and is $\Omega (\sqrt{n/\log{n}})$.

Balcan and Harvey~\cite{Balcan:2010} take the learning perspective to the study of the complexity of submodular functions. They study the problem of probabilistically learning a monotone submodular function, given the values the function takes on a polynomial sized sample of its domain. They provide a lower bound of $\Omega (n^{1/3})$ on the best possible approximation a learning algorithm can give to the submodular function, even when it knows the underlying sampling distribution, and the submodular function to be learned is Lipschitz. Another result with this perspective is by Balcan et al.~\cite{BalcanHarveyIwata:2012} who show that a symmetric non-monotone submodular function can be approximated to within $\sqrt{n}$ by the square root of a quadratic function. Furthermore, they show how to learn such submodular functions.

Badanidiyuru et al.~\cite{Badanidiyuru:2012}, motivated by the problem of communicating bidders' valuations in combinatorial auctions, study how well specific classes of set functions can be approximated given the constraint that the approximating function be representable in polynomially many bits. They named such an approximation a {\em sketch}, proving that coverage functions admit sketches with an arbitrarily good approximation factors. Additionally, for the larger class of monotone subadditive functions, they construct sketches that achieve an approximation of $\sqrt{n}\cdot polylog(n)$. Combining the results of Badanidiyuru et al.~\cite{Badanidiyuru:2012} and Balcan and Harvey~\cite{Balcan:2010}, a lower bound of $\Omega(n^{1/3}/\log^2{n})$ follows for the approximation of monotone submodular functions by budgeted additive functions (Appendix~\ref{sec:monotone-by-budgeted}).

Testing of submodular functions has been studied recently by Seshadri and Vondr{\'a}k~\cite{Seshadhri:2010} for general monotone submodular functions and for coverage functions by Chakrabarty and Huang~\cite{ChakrabartyH12}. The goal here is to query the function on few domain points and answer whether a function is \emph{close} to being submodular or not. The measure of closeness is the fraction of the domain in which the function needs to be modified so as to make it submodular.

\subsection{Preliminaries and Formal Statement of Results}
Given a ground set $U$, a function $f:2^{U}\rightarrow \RR_+$ is called \emph{submodular} if for all subsets $S,T\subseteq U$, we have
$f(S)+f(T)\geq f(S\cup T)+f(S\cap T)$.
A submodular function $f$ is called non-negative if $f(S)\geq 0$ for each $S\subseteq U$. In this paper, we only consider non-negative submodular functions. A submodular
function $f$ is called \emph{symmetric} if $f(S)=f(U\setminus S)$ for each $S\subseteq U$, and \emph{monotone} if $f(S)\leq f(T)$ for each $S\subseteq T\subseteq U$.
We say that a class of functions $\G$ {\em $\theta$-approximates} another class $\F$, if for every $f\in \F$ there exists a $g\in\G$ such that
$ f(S) \leq g(S) \leq \theta \cdot f(S)$ for any $S\subseteq U$. We denote by $n$ the size of the ground set $U$.
%
%
%
%
%
%
We now define certain subclasses of submodular functions that we consider in the paper.
\begin{definition}[Coverage function]\label{def:coverage}
A function $f$ is a {\em coverage function} if there exists an auxiliary ground set $Z$, a weight function $w:Z\rightarrow \RR_+$ and family of subsets $\left\{ A_i:A_i\subseteq Z,i\in U\right\}$ such that $\forall S\subseteq U$, $f(S) = \sum _{z\in \cup_{i\in S} A_i}w(z)$. 
\end{definition}

\begin{definition}[Budgeted additive function]\label{def:budgeted-additive}
A function $f$ is a {\em budgeted additive function} if there exist non-negative reals $a_i$ for each $i\in U$, and a non-negative real $B$ such that $\forall S\subseteq U$, $f(S)=\min\{B,\sum_{i\in S}a_i\}$.
\end{definition}
It is well known that coverage functions and budgeted additive functions are monotone submodular functions.

\begin{definition}[Cut function]\label{def:cut-function}
A function $f$ is a {\em directed cut function} if there exists a directed graph $G=(U,A)$ with non-negative arc
weights $w:A\rightarrow \RR_+$, such that $\forall S \subseteq U$, $f(S)=w(\delta^+(S))$, where $\delta^+(S)$ denotes the set of outgoing arcs, with their tails in $S$ and heads in $\bar{S}$, and $w(F) \triangleq \sum_{a\in F}w(a)$ for any subset $F\subseteq A$ of arcs.
%
\end{definition}
Similarly, one can define $f$ to be the {\em undirected cut function} of an undirected graph by substituting $\delta^+(S)$ with $\delta(S)$, the set of edges with exactly one endpoint in $S$.
It is well known that cut functions, whether directed or undirected, are submodular. Furthermore, clearly, undirected cut functions are symmetric.

Let us now formally state our main results:

\begin{theorem}\label{them:General_DirectedCUT}
Let $f:2^{U}\rightarrow \RR_+$ be a non-negative submodular function with $f(\emptyset)=f(U)=0$.
Then the class of directed cut functions $(n^2/4)$-approximates $f$.
Moreover, there exists a non-negative submodular function $f:2^{U}\rightarrow \RR_+$ with $f(\emptyset)=f(U)=0$ such that any directed cut function cannot approximate $f$ within a factor better than $n^2/4$.
\end{theorem}

\begin{theorem}\label{them:Symmetric_UndirectedCUT}
Let $f:2^{U}\rightarrow \RR_+$ be a non-negative symmetric submodular function with $f(\emptyset)=0$. Then the class of undirected cut functions $(n-1)$-approximates $f$.
Moreover, there exists a symmetric submodular function $f:2^{U}\rightarrow \RR_+$ with $f(\emptyset)=0$ such that any undirected cut function cannot approximate $f$ within a factor better than $n/4$.
\end{theorem}

\begin{theorem}\label{thrm:Bugdet_by_Coverage}
Let $f:2^{U}\rightarrow \RR_+$ be a budgeted additive function. Then coverage functions $(e/(e-1))$-approximates $f$.
Moreover, for every fixed $\varepsilon > 0$, there exists a budgeted additive function $f:2^{U}\rightarrow \RR_+$ such that any coverage function cannot approximate
$f$ within a factor better than $e/(e-1)-\varepsilon$.
\end{theorem}

Theorems~\ref{them:General_DirectedCUT}, \ref{them:Symmetric_UndirectedCUT} and \ref{thrm:Bugdet_by_Coverage}, are proved in Sections~\ref{sec:General_DirectedCUT},
~\ref{sec:Symmetric_UndirectedCUT} and \ref{sec:Budget_by_Coverage} respectively.

%% file: general_dicut.tex
\section{Approximating General Submodular Functions by Directed Cut Functions of Graphs}\label{sec:General_DirectedCUT}
In this section we prove Theorem~\ref{them:General_DirectedCUT} which provides a tight approximation of a non-negative submodular function $f$ using a directed cut function of a graph $G$.
Before proving the main result of this section, we first state a technical lemma, whose proof we defer to Appendix~\ref{sec:omitted-proofs}.
\begin{lemma}\label{lem:Submodular_Intersection}
For every submodular function $f$, and any collection of sets $A_{1}$, $A_{2}$, $\ldots$, $A_{n} \subseteq U$: $f(\cap_{i=1}^{n}A_{i}) \le \sum_{i=1}^{n}f(A_{i})$.
\end{lemma}
We are now ready to prove Theorem~\ref{them:General_DirectedCUT}.

\begin{proof}[Proof of Theorem \ref{them:General_DirectedCUT}]
{\bf \hspace*{\fill} \\Upper Bound:}
Given a submodular function $f$, we  construct a directed graph $G=(U,A)$ with non-negative weights $w$ on the arcs such for every $S\subseteq U$,
$ f(S)\leq w(\delta^+ (S)) \leq n^2/4\cdot f(S)$.
%
For every $(u,v) \in U \times U$ and $u\neq v$, introduce a directed arc from $u$ to $v$ with weight: $w_{uv} =  f(T_{uv}) $ where $ T_{uv}  = \text{argmin} \left\{f(R)~:~ R\subseteq U, u \in R, v \notin R\right\}$.
%
We start by proving that:
\begin{align}
 f(S) \le w(\delta^+(S)) & ~~~\forall S\subseteq U. \label{eq1}
\end{align}
If $S=U$ or $S=\phi$, then clearly $w(\delta^+(S))=f(S)=0$ and (\ref{eq1}) holds.
We now restrict our attention to the case where $S,\bar{S}\neq \emptyset$.
For any $u\in S$ note that $u \in \cap_{v \in \bar{S}}T_{uv}$, since the definition of $T_{uv}$ implies that $u\in T_{uv}$ for all $v\in \bar{S}$.
Additionally, for any $w\in \bar{S}$ note that $w \notin \cap_{v \in \bar{S}}T_{uv}$, since the definition of $T_{uv}$ implies that $w \notin T_{uw}$.
Thus, one can conclude that $\cup_{u \in S}\cap_{v\in \bar{S}}T_{uv}=S $ and therefore,
$$  f(S) \overset{\text{(\romannumeral 1)}}\le \sum_{u \in S}f(\cap_{v\in \bar{S}}T_{uv}) \overset{\text{(\romannumeral 2)}}\le \sum_{u \in S}\sum_{v \in \bar{S}}w_{uv}=w(\delta ^+(S)).$$
Inequality (\romannumeral 1) is derived from the fact that $f$ is submodular and non-negative. Inequality (\romannumeral 2) is derived from Lemma \ref{lem:Submodular_Intersection}
and the definition of $w_{uv}$.
This concludes the proof of (\ref{eq1}).

We continue by proving that:
\begin{align}
w(\delta ^+ (S)) \leq \frac{n^2}{4} f(S)~~~ \forall S\subseteq U. \label{eq2}
\end{align}
If $S=U$ or $S=\phi$, then clearly $w(\delta^+(S))=f(S)=0$ and (\ref{eq2}) holds.
We now restrict our attention to the case where $S,\bar{S}\neq \emptyset$.
Note that for any $u\in S$ and $v\in \bar{S}$, by the definition of $T_{uv}$: $f(T_{uv})\leq f(S)$.
Thus, one can conclude that: $$ w(\delta^+ (S)) \overset{\text{(\romannumeral 1)}}= \sum _{u\in S} \sum _{v\in \bar{S}}f(T_{uv}) \overset{\text{(\romannumeral 2)}}\leq \sum _{u\in S} \sum _{v\in \bar{S}}f(S)
\overset{\text{(\romannumeral 3)}}\leq \frac{n^2}{4}f(S).$$
Equality (\romannumeral 1) is by the definition of weights $w_{uv}$.
Inequality (\romannumeral 2) is be the definition of $T_{uv}$.
Inequality (\romannumeral 3) is by the fact that the number of pairs $(u,v)\in S\times \bar{S}$ is at most $n^2/4$.
This concludes the proof of (\ref{eq2}). Combining both (\ref{eq1}) and (\ref{eq2}) concludes the proof of the upper bound of the theorem.

\noindent {\bf Lower Bound:}
Assume that $n$ is even and fix an arbitrary $A\subseteq U$ of size $\vert A \vert =n/2$.
Consider the following function $f$
\begin{align*}
f(S)=\left\{
\begin{array}{rl}
1&\text{ if }S\cap A \ne \emptyset, \bar{A}\setminus S \ne \emptyset\\
0& \text{ otherwise }
\end{array}
\right.
\end{align*}
Namely, $f(S)$ is the indicator function that $S$ hits $A$ but does not hit all of $\bar{A}$. A simple check shows that $f$ is submodular.
Let $G=(U,A)$ be a weighted graph with non-negative weights $w:A\rightarrow \RR_+$ on the arcs whose directed cut function satisfies for each set $S\subseteq U$,
$f(S)\leq w(\delta^{+}(S))\leq \theta \cdot f(S)$
for some $\theta$. We will show that $\theta \geq \frac{n^2}{4}$ proving the lower bound.

First, we prove that the arcs with non-zero weight must go from $A$ to $\bar{A}$.
We consider the following cases.
\begin{enumerate}
\item Consider an edge $(u,v) \in A \times A$. But $(u,v)\in \delta^+({U\setminus \left\{v\right\}})$ and $w(\delta ^+(U\setminus \left\{v\right\})) \leq \theta \cdot f(U\setminus \left\{v\right\})=0$ since $\bar{A}\setminus (U\setminus \left\{v\right\})=\emptyset$. Thus, $w_{(u,v)}=0$.
\item Consider an edge $(u,v) \in \bar{A} \times \bar{A}$. But $(u,v)\in \delta^+(\bar{A}\setminus \{v\})$ and $w(\delta ^+(\bar{A}\setminus \left\{v\right\})) \leq \theta \cdot f(\bar{A}\setminus \left\{v\right\})=0$ since $A \cap (\bar{A}\setminus \left\{v\right\})=\emptyset$. Thus, $w_{(u,v)}=0$.
\item Consider an edge $(u,v) \in \bar{A} \times A$. But $(u,v)\in \delta^+(\bar{A})$ and , $w(\delta ^+(\bar{A}))\leq \theta \cdot f(\bar{A})=0$ since $A \cap \bar{A}=\emptyset$. Hence, $w_{(u,v)}=0$.
\end{enumerate}
Therefore, all arcs with non-zero weight must go from $A$ to $\bar{A}$. For any $u\in A$ and $v\in \bar{A}$ note that $w_{(u,v)}\geq 1$ since:
\begin{align}
w_{(u,v)}=w(\delta ^+\left(\left\{ u\right\}\cup \left( \bar{A}\setminus \left\{ v\right\}\right)\right)) \overset{\text{(\romannumeral 1)}}\geq   f \left(\left\{ u\right\}\cup \left( \bar{A}\setminus \left\{ v\right\}\right)\right) \overset{\text{(\romannumeral 2)}}= 1. \label{eq3}
\end{align}
Inequality (\romannumeral 1) is derived from the fact $w(\delta^{+}(S))\geq f(S)$ for each set $S$. Equality (\romannumeral 2) is by the definition of $f$.
Furthermore, note that:
\begin{align}
 \frac{n^2}{4}=\vert A\vert\cdot \vert\bar{A}\vert \overset{\text{(\romannumeral 1)}} \leq w(\delta ^+(A)) \overset{\text{(\romannumeral 2)}}\leq \theta\cdot f(A) \overset{\text{(\romannumeral 3)}}= \theta .\label{eq4}
\end{align}
Inequality (\romannumeral 1) is derived from inequality (\ref{eq3}).
Inequality (\romannumeral 2) is derived from the fact that $w(\delta ^+(S))\leq \theta f(S)$ for each set $S\subseteq U$.
Equality (\romannumeral 3) is by the definition of $f$.
Note that inequality (\ref{eq4}) implies that $\theta \geq \frac{n^2}{4}$, thus, concluding the proof of the lower bound of the theorem.
\end{proof}

%% file: symmetric_cut.tex
\section{Approximating Symmetric Submodular Functions by Undirected Cut Functions of Graphs}\label{sec:Symmetric_UndirectedCUT}
In this section, we prove Theorem~\ref{them:Symmetric_UndirectedCUT} which provides upper and lower bounds on the approximation of a symmetric submodular function using an undirected cut function of a graph.
For the upper bound, our algorithm uses Gomory-Hu trees of symmetric submodular functions \cite{GomoryH61,Queyranne93}. Given a symmetric non-negative submodular function $f$, a tree $T=(U,E_T)$ is a Gomory-Hu tree if for every edge $e=(u,v)\in E_T$: $ f(R_e) = \min \left\{ f(R)~:~ R\subseteq U, u\in R, v\notin R\right\}$,
where $R_e$ is one of the two connected components obtained after removing $e$ from $T$ (since $f$ is symmetric, it does not matter which one of the two connected components we choose). In other words, in a Gomory-Hu tree, the cut $e=(u,v)$ induced in $T$, corresponds to a minimum value subset that separates $u$ and $v$. We prove that the cut function of the Gomory-Hu tree of $f$ is a good approximation.

\begin{proof} [Proof of Theorem \ref{them:Symmetric_UndirectedCUT}]
{\bf \hspace*{\fill} \\ Upper Bound:}
Let $f$ be a symmetric submodular function. We shall construct an undirected tree $T=(U,E_T)$ with non-negative weights $w:E\rightarrow \RR_+$ on the edges such that
for every $S\subseteq U$, $ f(S) \leq w(\delta(S)) \leq (n-1)\cdot f(S)$.
%
We set $T$ to be a Gomory-Hu tree of $f$ and let the weight of any edge $e=\{u,v\}$ to be $f(R_e)$ where $R_e$ is the one of the two connected components obtained after removing edge $e$. As mentioned above, the weight of edge $e=\{u,v\}$ is the minimum of $f(R)$ over all $R$ separating $u$ and $v$.

Fix an arbitrary $S\subseteq U$ and denote by $\left\{ e_1,\ldots,e_k\right\}$ all the edges crossing the cut that $S$ defines in $T$.
Let $T_1,\ldots,T_{k+1}$ denote the partition of $U$ induced by deleting the edges $e_1,\ldots,e_k$ from $T$. Furthermore, denote by $\{S_1,\ldots,S_p\}$ the non-empty sets in $\{T_i\cap S:1\leq i\leq k+1\}$. Observe that $S_1,\ldots,S_p$ is a partition of $S$.
Since each $e_i$, $1\leq i\leq k$, has exactly one vertex in $S$ and the other in $\bar{S}$, we can associate $e_i$ with a {\em unique}
set from $S_1,\ldots,S_p$, the set containing one of the endpoints of $e$. Additionally, let us denote $F_i$ to be the edges which are associated
 with set $S_i$ for each $1\leq i\leq p$. Clearly $F_1,\ldots,F_p$ form a partition of $\{e_1,\ldots, e_k\}$.

We claim that for every $1\leq i\leq p$:
\begin{align}
f(S_i) \leq \sum _{f\in F_i}f(R_{f}). \label{eq5}
\end{align}
Recall that $R_f$ is one of the connected component after removing edge $f$ from $T$. Since $S_i$ is a subset of a connected component formed after removing all the edges $\{e_1,\ldots,e_k\}$ from $T$, it must be contained in one of the components formed after removing edge $f\in F_i$ from $T$. Without loss of generality, we assume that $R_{f}\cap S_i=\emptyset$ for each edge $f\in F_i$. It is straightforward to see that $\cap_{f\in F_i} \bar{R}_f =S_i$. Now, we have
$$ \sum _{f\in F_i}f(R_f)  \overset{\text{(\romannumeral 1)}}\geq f\left( \cup _{e\in F_i}R_{f}\right)
  \overset{\text{(\romannumeral 2)}}= f\left( \overline{\cup _{f\in F_i}R_{f}}\right)
  = f\left( \cap _{e\in F_i}\bar{R}_{f}\right)
  \overset{\text{(\romannumeral 3)}}= f(S_i). $$
Inequality (\romannumeral 1) is derived from the fact that $f$ is submodular and non-negative.
Equality (\romannumeral 2) is derived from the symmetry of $f$.
Equality (\romannumeral 3) is derived from the fact that $\cap_{f\in F_i} \bar{R}_f =S_i$.

We start by proving that:
\begin{align}
 f(S) \leq w(\delta(S)) ~~~\forall S \subseteq U. \label{eq6}
\end{align}
This can be proved as follows:
$$ w(\delta(S)) \overset{\text{(\romannumeral 1)}}=  \sum _{i=1}^k f(R_{e_i})
\overset{\text{(\romannumeral 2)}}= \sum _{i=1}^p \sum _{f\in F_i} f(R_{f})
\overset{\text{(\romannumeral 3)}}\geq \sum _{i=1}^p f(S_i)
\overset{\text{(\romannumeral 4)}}\geq f\left(\cup _{i=1}^p S_i\right) =  f(S).$$
Equality (\romannumeral 1) is by the definition of edge weights in $T$.
Equality (\romannumeral 2) is by the fact that $F_1,\ldots,F_p$ form a partition of $\{e_1,\ldots, e_k\}$.
Inequality (\romannumeral 3) is derived from inequality (\ref{eq5}).
Inequality (\romannumeral 4) is derived from the fact that $f$ is submodular and non-negative.
This concludes the proof of (\ref{eq6}).

We continue by proving that:
\begin{align}
w(\delta(S)) \leq (n-1)\cdot f(S) ~~~\forall S \subseteq U. \label{eq7}
\end{align}
Let $u$ and $v$ be the endpoints of edge $e_i$ and without loss of generality assume that $u\in S$ and $v\notin S$.
Note that for every $1\leq i\leq k$, $f(R_{e_i})\leq f(S)$ since $S$ is a candidate set separating $u$ and $v$.
Hence, one can conclude that:
$$ w(\delta(S)) \overset{\text{(\romannumeral 1)}}=  \sum _{i=1}^k f(R_{e_i}) \overset{\text{(\romannumeral 2)}}\leq k \cdot f(S)
\overset{\text{(\romannumeral 3)}}\leq (n-1)\cdot f(S).$$
Equality (\romannumeral 1) is by the definition of edge weights in $T$.
Inequality (\romannumeral 2) is what we proved above, and inequality (\romannumeral 3) is derived from the fact that $T$ contains at most $n-1$ edges, thus, $k\leq n-1$.
This concludes the proof of (\ref{eq7}). Combining both (\ref{eq6}) and (\ref{eq7}) concludes the proof of the upper bound of the theorem.

\noindent {\bf{Lower Bound:}} Refer to Appendix~\ref{sec:LowerBound_Symmetric}. 
\end{proof}

%% file: budget_coverage.tex
\section{Approximating Budgeted Additive Functions by Coverage Functions}\label{sec:Budget_by_Coverage}
In this section, we present matching upper and lower bounds for approximating budgeted additive functions by coverage functions (Theorem~\ref{thrm:Bugdet_by_Coverage})\footnote{It is easy to show that a coverage function can be written exactly as a sum of budgeted additive functions.}.
The following lemma from Chakrabarty and Huang~\cite{ChakrabartyH12} provides the alternate representation of coverage functions used in our proof of the lower bound.
\begin{lemma}\label{lem:charac-coverage}\cite{ChakrabartyH12}
A function $f:2^{U}\rightarrow \RR_+$ is a coverage function if and only if there exist reals $x_T\geq 0$ for each $T\subseteq U$ such that $f(S)=\sum_{T:T\cap S\neq \emptyset} x_T$ for each $S\subseteq U$.
\end{lemma}

\begin{proof}[Proof of Theorem~\ref{thrm:Bugdet_by_Coverage}]
{\bf \hspace*{\fill} \\Upper Bound:} Refer to Appendix~\ref{sec:UpperBound_Budget}.\\
\noindent {\bf Lower Bound:}
We will construct a budgeted additive function which cannot be approximated by coverage functions to factor better than $\frac{e}{e-1}-\epsilon$ for any $\epsilon>0$. We will consider the family of budgeted additive function $f_{k}$, parameterized by the size of domain $|U|=n$ they are defined on, where $n=k^{2}$ for some integer $k$. Under $f_{k}$, all $n=k^{2}$ items have value one and the budget is $k$. Please note that these also constitute a family of uniform matroid rank functions. Therefore,
\vspace{-.15cm}
\begin{align}
f_{k}(S)=\left\{
\begin{array}{ll}
\vert S \vert &  \text{if } \vert S \vert \le k\\
k & \text{o.w.}
\end{array}
\right.
\end{align}

Let $h_{k}$ be a coverage function that gives the maximum value of $\beta$ such that $\forall S\subset [n],~~\beta \cdot f_{k}(S) \le h_{k}(S) \le f_{k}(S)$ and $\alpha_{k}$ be the value of $\beta$ as given by $h_{k}$. Observe that here function $h_k$ is always smaller than the function $f_k$. The function $\frac{h_k}{\beta}$ would give a $\frac{1}{\beta}$-approximation for approximating function $f_k$. This slight change in notation helps for exposition below. We shall show that as $k\to \infty$, $\alpha_{k}$ tends to a value that is at most $1-1/e$. This shall prove our claim.

Using Lemma~\ref{lem:charac-coverage}, we note that $\alpha_k$ can be characterized by a solution to a linear problem $(P)$ given below. Here, the variables are $x_T$, one for each set $T\subseteq U$. The dual (D) of this linear program is given alongside. We will construct a dual solution of value approaching $1-\frac{1}{e}$ as $k\to \infty$. Since every feasible dual solution is an upper bound on $\alpha_{k}$, the result follows.

\newsavebox\primalprogram
\begin{lrbox}{\primalprogram}
\begin{minipage}{0.3\textwidth}
\begin{align*}
&\max \alpha_{k}&\text{(P)}\notag\\
\text{subject to}\notag\\
\forall S \subseteq U, ~~~& \sum_{T \cap S \ne \phi}x_{T} \le f_{k}(S) \notag\\
\forall S \subseteq U, ~~~& \alpha_{k} f_{k}(S) - \sum_{T \cap S \ne \phi}x_{T} \le 0\notag\\
\forall S \subseteq U, ~~~& x_{S} \ge 0\notag
\end{align*}
\end{minipage}
\end{lrbox}

\newsavebox\dualprogram
\begin{lrbox}{\dualprogram}
\begin{minipage}{0.3\textwidth}
\begin{align*}
& \min \sum_{S \subseteq U}f_{k}(S) \cdot u_{S}&\text{(D)}\notag\\
\text{subject to}\notag\\
\forall S \subseteq U, & \sum_{T \cap S \ne \phi} (u_{T}-v_{T}) \ge 0\notag\\
& \sum_{S \subseteq U}f_{k}(S) \cdot v_{S}\ge 1
\end{align*}
\end{minipage}
\end{lrbox}

\begin{table}[h]
\begin{center}
\begin{tabular}{c|c}
\usebox{\primalprogram}&\usebox{\dualprogram}
\end{tabular}
\end{center}
\end{table}

Since, $f(\cdot)$ is symmetric across sets of the same cardinality, we can assume, without loss of generality, that the optimal dual solution is also symmetric. Specifically, the values of the dual variables $u_{T}$ and $v_{T}$ shall depend only the cardinality $\vert T \vert$. Let us write the symmetrized dual program.

\begin{align}
& \min~~ \sum_{j=1}^{k} j \cdot {n \choose j} \cdot u_{j} + \sum_{j=k+1}^{n}k \cdot {n \choose j} \cdot u_{j}&\text{Symmetrized Dual Program}\notag\\
\textrm{subject to}\notag\\
\forall j \in [n], ~~~& \sum_{i=1}^{n}  \left( {n \choose i}-{n-j \choose i} \right)(u_{i}-v_{i}) \ge 0 \label{eq:dual-sym-uv}\\
& \sum_{j =1}^{k} j \cdot {n \choose j} \cdot v_{j} + \sum_{j=k+1}^{n} k \cdot {n \choose j} \cdot v_{j}\ge 1 \label{eq:dual-sym-v-only}
\end{align}

Let $c_{j}$ denote the coefficient of $v_{k}$ in the equation corresponding to set size $j$, i.e., $c_{j}={n \choose k} - {n-j \choose k}$. Further, define $\Delta c_{j} = c_{j+1}-c_{j}$.

 We give the following solution to the dual linear program. Let $v_{k} = \frac{1}{{n \choose k} \cdot k }$, $u_{1}=  \Delta c_{k} \cdot v_{k}$ and $u_{n}= (c_{k} - k \cdot \Delta c_{k}) \cdot v_{k}$. Rest of the variables are set to zero.

We first show that the above solution is feasible for the dual and has objective value that tends to $1-1/e$ as $k \to \infty$.
It is easy to see that with the proposed setting of $v_{k}$, Equation~(\ref{eq:dual-sym-v-only}) is satisfied. To show that Equation~(\ref{eq:dual-sym-uv}) is satisfied, we show that $\forall j \in [n]$, $j \cdot u_{1} + u_{n} \ge ({n \choose k} - {n-j \choose k})v_{k}$. Using our notation, it suffices to show that for all $j\in [n]$, $(j-k) \cdot \Delta c_{k} + c_{k} \ge c_{j}$.

\begin{claim}
$c_{j}$ is an increasing function of $j$ and $\Delta c_{j} = c_{j+1}-c_{j}$ is a decreasing function of $j$.
\end{claim}
\begin{proof}
 $ \Delta c_{j}=c_{j+1}-c_{j}= ({n \choose k} - {n-j-1 \choose k}) -  ({n \choose k} - {n-j \choose k})={n-j \choose k} - {n-j-1 \choose k}={n-j-1 \choose k-1}$.
\end{proof}
\begin{claim}
For all $j\in [n]$, $(j-k) \cdot \Delta c_{k} + c_{k} \ge c_{j}$.
\end{claim}
\begin{proof}
%
\begin{itemize}
\item For $j=k+i$ such that $0\le i \le n-k$: $\text{LHS = } i \cdot \Delta c_{k} + c_{k}  \ge  c_{k+i} = \text{RHS}$
\item For $j=k-i$ with $0\le i \le k$: $\text{LHS = } c_{k} - i \cdot \Delta c_{k} \ge c_{k}-\sum_{q=k-i}^{k-1}\Delta c_{q} = c_{k-i} \cdot v_{k} = \text{RHS}$
\end{itemize}
where the second inequality in both the cases follows because $\Delta c_{j}$ is a decreasing function in $j$.
\end{proof}

Let us now bound the value of the dual objective function. Look at the value that the dual objective function attains with this setting of variables. The function value is $n \cdot u_{1} + k \cdot u_{n} = (n \cdot \Delta c_{k} + k \cdot (c_{k} - k \cdot \Delta c_{k})) \cdot v_{k}$. Since $n=k^{2}$, the dual objective value is equal to $$k \cdot c_{k} \cdot v_{k}= 1 - \frac{{k^{2}-k \choose k}}{{k^{2} \choose k}}.$$ This quantity tends to $1-1/e$ as $k \to \infty$.
\end{proof}

%% file: future_directions.tex
\section{Future Directions}
We mention here a couple of main research questions that are left open by the present work. The first is how well a non-negative monotone submodular function can be approximated by the sum of matroid rank functions. Dughmi et al. \cite{DughmiRY11} show that the Hessian matrix for a matroid rank sum has to be negative semi-definite, and it is easy to come with a budgeted additive function that does not obey this property. Hence, we cannot hope for the best approximation factor for a submodular function by a matroid rank sum to be 1; in fact we can show that the approximation factor cannot be better than some constant bounded away from 1. In terms of positive results, a $O(\sqrt(n))$ factor approximation follows from \cite{Goemans:2009} and a $O(\frac{\max_{e\in U}f(e)}{\min_{e\in U}f(e)})$ follows from a result in Section 44.6(B) in \cite{schrijver:2003}.

The second is approximating a non-negative symmetric submodular function by a hypergraph cut function (in this paper, we only considered graph cut functions). The lower bound example in the paper for graph cut functions can be extended to show that a $r$-regular hypergraph cannot approximate to a factor better than $O(\frac{n}{r})$. In terms of positive results, we know no better than the ones mentioned in this paper.

%% file: appendix.tex
\def \I {\mathcal{I}}
\def \M {\mathcal{M}}

\section{Lower Bound of Theorem~\ref{them:Symmetric_UndirectedCUT}}\label{sec:LowerBound_Symmetric}
\begin{proof}[Lower Bound of Theorem~\ref{them:Symmetric_UndirectedCUT}]
Consider the following symmetric submodular function $f$:
\begin{align*}
f(S)=\left\{
\begin{array}{rl}
1& \;\; \text{if } S\ne \emptyset,U\\
0& \;\; \text{ otherwise }
\end{array}
\right.
\end{align*}
Let $G=(U,E)$ be an edge weighted graph with non-negative weights $w:E\rightarrow \RR_+$ on the edges whose cut function satisfies
 $f(S)\leq w(\delta(S))\leq \theta \cdot f(S)$
 for each set $S\subseteq U$ for some $\theta$. We will show that $\theta \geq \frac{n}{4}$.

For any vertex $v\in U$,
$1=f(\{v\})\leq w(\delta(\{v\}))$.
Thus, the total weight of edges in $G$ is at least $\frac{1}{2}\sum_{v\in U}w(\delta \left( \left\{ v\right\}\right))\geq  \frac{n}{2}$.
Every undirected graph has a non-trivial cut that contains at least half the total weight of edges in the graph, thus, there exists a cut $S\subseteq U$, $S\neq \emptyset,U$, where $w(\delta(S)) \geq \frac{n}{4}$. The existence of such a cut can be shown by picking a cut at random where each vertex is in $S$ with probability $\frac{1}{2}$ independently. The expected weight of the cut will be exactly half the total weight of all edges. Now, we have $ \frac{n}{4}\leq w(\delta(S)) \leq \theta \cdot f(S) =\theta$,
concluding the proof of the lower bound of the theorem. 
\end{proof}

\section{Upper Bound of Theorem~\ref{thrm:Bugdet_by_Coverage}}\label{sec:UpperBound_Budget}
\begin{proof}[Upper Bound of Theorem~\ref{thrm:Bugdet_by_Coverage}]
Consider any budgeted additive function $f(\cdot)$ over some domain $U$, with budget $B$ and the values of the elements be denoted by $v_{1}, v_{2}, \cdots v_{n}$ where $n=|U|$. Without loss of generality, we assume all these values to be integers.
Take an auxiliary ground set $G$ of size $B$. For each $i\in U$, construct a set $A_{i} \subseteq G$, formed by choosing $v_{i}$ points (with replacement) at random from $G$. Consider function $g:2^{U} \rightarrow \mathbb{Z}$, defined as $g(S) = \vert\cup_{i\in S}A_{i}\vert$ for all $S\subseteq U$.

By definition, $g(\cdot)$ is a coverage function. Furthermore, it is easy to see that for all $S\subseteq U$, $g(S) \le f(S)$. We now show that $\mathbb{E}[g(S)] \ge (1-1/e) \cdot f(S)$, where the expectation is taken over the randomness of the procedure described to construct $g(S)$. Note that $g'(\cdot)=\mathbb{E}[g(\cdot)]$ is a coverage function.
Consider any set $S \subseteq U$. Let $f(S) = V$, i.e., $\sum_{i\in S}v_{i}=V$. Consider the case when $V<B$. Consider any point in auxiliary ground set $G$. The probability that this point is not covered by any of the sets $A_{i}$ for $i \in S$ is at most $ (1-1/B)^V$. Hence, the expected value of $\vert\cup_{i \in S}A_{i}\vert$ is at least $B \cdot (1-(1-1/B)^{V})\ge B \cdot (1-e^{-V/B}) \ge (1-1/e) \cdot V$. Here we use the inequality $1-e^{-x} \ge (1-1/e)\cdot x$. Hence, $\vert\cup_{i \in S}A_{i}\vert \ge (1-1/e) \cdot f(S)$. The proof for the case when $V=B$ is similar.
Thus for each set $S\subseteq U$, we have

\[(1-\frac{1}{e})f(S)\leq g'(S)\leq f(S) .\]
Thus, we obtain the function $\frac{e}{e-1}g'(\cdot)$ approximates $f$ within a factor of $\frac{e}{e-1}$. 
\end{proof}

\section{Proof of Lemma~\ref{lem:Submodular_Intersection}}
\label{sec:omitted-proofs}
\begin{proof}[Proof of Lemma~\ref{lem:Submodular_Intersection}]
The following inequalities are derived from the definition of submodularity:

\begin{align*}
f(A_{1}) + f(A_{2}) &\ge f(A_{1} \cap A_{2}) + f(A_{1} \cup A_{2})\\
f(A_{3}) + f(A_{1}\cap A_{2}) &\ge f(A_{1} \cap A_{2} \cap A_{3}) + f((A_{1} \cap A_{2})\cup A_{3})\\
&~\vdots\\
f(A_{n}) + f(\cap_{i=1}^{n-1}A_{i}) &\ge f(\cap_{i=1}^{n}A_{i}) + f((\cap_{i=1}^{n-1}A_{i})\cup A_{n})
\end{align*}

Summing up the above inequalities and canceling common terms on the two sides and using the fact that $f$ is non-negative we obtain that
$\sum_{i=1}^{n}f(A_{i}) \ge f(\cap_{i=1}^{n}A_{i})$.
\end{proof}

\section{Uniform Submodular and Matroid Rank Functions}
\begin{definition}[Uniform Submodular Function]
A submodular function is said to be \emph{uniform} if the value it takes on a set depends only on the cardinality of the set.
\end{definition}
\begin{lemma}
\label{lem:uniform-matroid-and-concave}
Any non-negative, integer-valued, monotone, uniform, submodular function is $1$-approximated by a sum of uniform matroid rank functions, and hence by a sum of budgeted additive functions.
\end{lemma}
\begin{proof}
Consider an integer-valued, non-negative, monotone, uniform, submodular function $f(\cdot)$ over the universe $[n]$ and let $f_{k}$ be the value $f$ takes for sets $S$ of cardinality $k$. Consider uniform matroid rank functions $g_{1}, g_{2}, \cdots g_{n}$ where
\begin{align}
g_{i}(S) = \left\{ \begin{array}{ll}
\vert S \vert & \vert S \vert \le i\\
i & \vert S \vert > i
\end{array}\right.
\end{align}

We claim that there exist a set of $\alpha_{i}$'s, such that $\alpha_{i}\ge 0$ for all $i \in [n]$ such that
\begin{equation}
\forall j \in [n], f_{j} = \sum_{i=1}^{k} \alpha_{i} \cdot j + \sum_{i=k+1}^{n} \alpha_{i} \cdot i
\end{equation}
It is easy to see that the above claim implies that $f(S) = \sum_{i} \alpha_{i}\cdot g_{i}(S)$ for all $S \subseteq [n]$.

Now, we prove the claim. If, for every $j \in [n-1]$, we substract equation $j$ from $j+1$, we get the following set of equations
\begin{equation}
\forall j\in [n-1], f_{j+1} - f_{j} = \sum_{i=j+1}^{n}\alpha_{i}
\end{equation}

From here, we can see that the following assignments to $\alpha_{i}$ is a valid solution to the above equations. Set $\alpha_{1}=2\cdot f_{2} - f_{1}, \alpha_{n}=f_{n}- f_{n-1}$ and for $i \notin \{1,n\}$, set $\alpha_{i} = 2 \cdot f_{i} - f_{i+1} -f_{i-1}$. All the $\alpha_{i}$'s are positive since $f$ is monotone and $f$ is submodular.

Finally, it is easy to see that every uniform matroid rank function is also a budgeted additive function with all elements having value one, and the budget equal to the rank of the matroid. 
\end{proof}

%
%
%
%

\section{Application to Online Submodular Function Maximization}\label{sec:application}

We consider the problem of online submodular function maximization as studied by Buchbinder et. al~\cite{BuchbinderNRS12}. We are given a universe $U$ and matroid $\M=(U,\I)$. In an online manner, at each step for $1\leq i \leq m$, we are given a monotone submodular function $f_i:2^{U}\rightarrow \RR_+$. The goal is to maintain an independent set $F_i\in \I$ at any step $i$ such that $F_i\subseteq F_{i+1}$. The objective value to maximize is $\sum_{i=1}^m f_i(F_i)$. As in the notion of competitive analysis, any algorithm is compared to the best offline optimum $\max_{O\in \I} \sum_{i=1}^m f_i(O)$.

Buchbinder et. al~\cite{BuchbinderNRS12} give a $O(\log^2 n\log m\log f_{ratio})$-competitive algorithm when each of the submodular function is weighted matroid rank function where $f_{ratio}=\frac{\max_{i,a}f_i(\{a\})}{\min_{i,a:f_i(\{a\})>0}f_i(\{a\})}$. In particular, the result applies when each of the functions $f_i$ is a coverage function.

Using the fact every monotone submodular function can be approximated by a coverage function to a factor of $O(\sqrt{n}\log n)$, we directly obtain the following corollary.

\begin{corollary}
There is a $O(\sqrt{n}\log^3 n\log m f_{ratio})$-competitive online algorithm for the online submodular function maximization problem when each of the submodular functions is an arbitrary monotone submodular function.
\end{corollary}

\input{monotone_coverage.tex}

%% file: monotone_coverage.tex
\section{Approximating Monotone Submodular Functions by Coverage Functions and by Budgeted Additive Functions}
\label{sec:monotone-by-budgeted}
The two main results of the section are the following.
\begin{theorem}\label{thrm:Monotone_by_Budget}
Coverage functions can approximate every non-negative monotone submodular function to within a factor $O\left(\sqrt{n}\log{n}\right)$.
Additionally, the class of coverage functions cannot approximate every non-negative monotone submodular function to a factor within $o\left( \frac{n^{1/3}}{\log^2{n}}\right)$.
\end{theorem}

\begin{theorem}
The class of sum of budgeted additive functions can approximate every non-negative monotone submodular function to a factor within $\left( \sqrt{n}\log{n}\right)$.
Additionally, the class of sum of budgeted additive functions cannot approximate every non-negative monotone submodular function to a factor within $o \left( \frac{n^{1/3}}{\log^2{n}}\right)$.
\end{theorem}

\subsection{Upper Bound}
For the upper-bound, we show that budgeted additive functions can $\sqrt{n}\log(n)$-approximate the class of non-negative, monotone, submodular functions. Then we use Theorem~\ref{thrm:Bugdet_by_Coverage} to infer that the coverage functions  too can give approximately the same guarantee.
\begin{lemma}
\label{lem:Budget-approx-monotone}
The class of sum of budgeted additive functions can approximate every non-negative monotone submodular function to within a factor $\left(\sqrt{n}\log{n}\right)$.
\end{lemma}

The following corollary follows easily from Lemma~\ref{lem:Budget-approx-monotone} and Theorem~\ref{thrm:Bugdet_by_Coverage}
\begin{corollary}
\label{crl:Coverage-approx-monotone}
The class of sum of coverage functions can approximate every non-negative monotone submodular function to factor $O\left(\sqrt{n}\log{n}\right)$.
\end{corollary}

The proof of Lemma~\ref{lem:Budget-approx-monotone} follows from Lemmas~\ref{lem:goemans} and~\ref{lem:budget-additive-weighted-concave}. Lemma~\ref{lem:goemans} \cite{Goemans:2009} gives a particular function that $\sqrt{n}\log(n)$-approximates a general monotone, non-negative, sub-modular function, and Lemma~\ref{lem:budget-additive-weighted-concave} implies that this approximating function can be written as a sum of budgeted additive functions.

\begin{lemma}\cite{Goemans:2009}
\label{lem:goemans}
For every monotone submodular function $f:2^{U}\rightarrow R_+$, there exists positive reals $a_e$ for each $e\in E$ such that $g:2^{U}\rightarrow R_+$ defined as $g(S)=\sqrt{\sum_{e\in S} a_e}$, approximates $f$ within factor $\sqrt{n}\log(n)$.
\end{lemma}

\begin{lemma}
\label{lem:budget-additive-weighted-concave}
Every submodular function $f:2^{[n]}\rightarrow \mathbb{Z}^{+}$  of the form $f(S)=g(\sum_{i \in S} a_{i})$, where $a_{i} \in \mathbb{Z}^{+}$ and $g$ is a non-negative, monotone, concave and integer valued on integral inputs, can be written as a sum of budgeted additive functions.
\end{lemma}
\begin{proof}
Let $m=\sum_{i=1}^{n}a_{i}$. Consider the function $h$, over the domain $[m]$, defined as $h(S)=g(\vert S \vert)$ for all $S \subseteq [m]$. Since $g(\cdot)$ is a non-negative, monotone, concave function, it is easy to verify that $h(\cdot)$ is a non-negative, monotone, submodular function. Construct $n$ mutually disjoint sets $A_{i} \subseteq [m]$ such that $\vert A_{i} \vert =a_{i}$.  Clearly, $\forall S \subseteq [n], f(S) = h(\cup_{i \in S}A_{i})$.

From Lemma~\ref{lem:uniform-matroid-and-concave}, we know that $h(\cdot)$ can be expressed as $\sum_{i=1}^{m}\alpha_{i} \cdot t_{i}(\cdot)$ where each $t_{i}$ is a uniform matroid rank function with rank $i$, over the domain $[m]$ and each $\alpha_{i}\ge0$.

This implies that for all $S\subseteq [n]$, $f(S) = \sum_{i=1}^{m}\alpha_{i} \cdot t_{i}(\cup_{i \in S}A_{i})$. Now, for every $i\in [m]$, construct the budget additive function $t'_{i}$, defined as $\forall S\subseteq [n], t_{i}(S)=\min\{\sum_{i\in S}v_{ij}, B_{i}\}$, where for $j\in [n]$, the value $v_{ij}=a_{j}$ and the budget $B_{i}$ is $i$. Since $t_{i}$ is a uniform matroid rank function of rank $i$ and $A_{i}$'s are mutually disjoint, we have for all $i\in [m]$,
\begin{equation}
\forall S \subseteq [n], t'_{i}(S) = t_{i}(\cup_{i \in S}A_{i})
\end{equation}

Therefore, we get for all sets $S \subseteq [n]$, $f(S) = \sum_{i=1}^{m}t'_{i}(S)$. 
\end{proof}

\subsection{Lower Bound}
For the lower bound, we first show that sum of coverage functions cannot approximate the class of monotone, submodular functions well, and then use Theorem~\ref{thrm:Bugdet_by_Coverage}, to infer that, therefore, even the class of sum of budgeted additive functions cannot approximate a monotone submodular function well.

\begin{lemma}
\label{lem:coverage-cannot-appx-montone-well}
The class of sum of coverage functions cannot approximate every non-negative monotone submodular function to a factor within $o\left( \frac{n^{1/3}}{\log^2{n}}\right)$.
\end{lemma}

An easy corollary of Lemma~\ref{lem:coverage-cannot-appx-montone-well} that follows from Theorem~\ref{thrm:Bugdet_by_Coverage} is the following.
\begin{corollary}
\label{lem:budget-cannot-appx-montone-well}
The class of sum of budgeted additive functions cannot approximate every non-negative monotone submodular function to a factor within $o \left( \frac{n^{1/3}}{\log^2{n}}\right)$.
\end{corollary}

We now present the proof of Lemma~\ref{lem:coverage-cannot-appx-montone-well}. We will need to use results from~\cite{Badanidiyuru:2012} and~\cite{Balcan:2010}, for which we first present a definition.

\begin{definition}
A $\beta$-sketch of a function $f:2^{U}\rightarrow \mathbb{R}$ is a polynomially sized (in $\vert U \vert$ and $1/(1-\beta)$) representable function $g$ such that $\forall S \subseteq U$, $\beta \cdot f(S) \le g(S) \le f(S)$.
\end{definition}
The following result is from~\cite{Badanidiyuru:2012}.
\begin{lemma}\cite{Badanidiyuru:2012}
\label{lem:coverage-good-sketch}
Coverage functions allow from arbitrary well sketches i.e., for any $\epsilon >0$, there exists a $1-\epsilon$ sketch.
\end{lemma}

The following result is from~\cite{Balcan:2010}. It gives a `large' family of matroid rank functions, such that any two functions in the class have at least one point where the values that they take differ by a `significant' factor.
\begin{lemma}\cite{Balcan:2010}
\label{lem:balcan-harvey}
For any $k=2^{o(n^{1/3})}$, there exists a family of sets $\mathcal{A}\subseteq 2^{[n]}$ with $\vert \mathcal{A} \vert=k$ and a family of matroids $\mathcal{M}=\{M_{B}\vert B \subseteq \mathcal{A}\}$ such that for all $B \subseteq \mathcal{A}$, it is the case that
\begin{align}
\forall S \in \mathcal{A}, r_{M_{B}}(S) = \left\{
\begin{array}{ll}
8 \log k & \text{if } S \in B\\
n^{1/3} & \text{if }  S \notin B
\end{array}
\right.
\end{align}
where $r_{M_{B}}$ is the rank function of the matroid $M_{B}$
\end{lemma}

\begin{proof}[of Lemma~\ref{lem:coverage-cannot-appx-montone-well}]
Let the class of matroid rank functions on the domain of size $n$ be $\alpha$-approximable by coverage functions, for some $\alpha$. That is, for a domain $[n]$, for all matroid rank function $r$, there exists a coverage function $g$ such that $\forall S \subseteq [n], r(S) \le g(S) \le \alpha \cdot r(S)$.

By Lemma~\ref{lem:coverage-good-sketch}, for every $\epsilon>0$ and every coverage function $g$, there exists a polynomially sized (polynomial in $n$ and $1/\epsilon$) representable function $h$ such that $\forall S\subseteq [n], (1-\epsilon)\cdot g(S) \le h(S) \le g(S)$. Hence, for all $\epsilon >0$ and for all matroid rank functions $r$, there exists a polynomial sized representable function $h$ such that $\forall S \subseteq [n], r(S) \le h(S)/(1-\epsilon) \le \frac{\alpha}{1-\epsilon} \cdot r(S)$. For any given $\epsilon>0$, there are only $2^{O(n,1/\epsilon)}$ many different $h$ functions.

From Lemma~\ref{lem:balcan-harvey}, for $k=2^{\log^{2}(n)}$, there exists family of sets $\mathcal{A} \subseteq 2^{[n]}$ with $\vert \mathcal{A}\vert=k$, and a $2^{k}$ sized matroid family $\mathcal{M}_{B}$ such that for all sets $A \in \mathcal{A}$ and $\forall B \subseteq \mathcal{A}$,
\begin{align}
\forall S \in \mathcal{A}, r_{M_{B}}(S) = \left\{
\begin{array}{ll}
8 \log^{2}n & \text{if } S \in B\\
n^{1/3} & \text{if }  S \notin B
\end{array}
\right.
\end{align}

Now while the number of different $g'$ functions are $2^{O(n,1/\epsilon)}$, the number of different matroid rank functions in this family is $2^{n^{\log(n)}}$. Hence, by pigeon-hole principle, there must be two matroids $B$ and $B'$ ($B \ne B'$) such that the best coverage functions $g$ and $g'$ approximating $B$ and $B'$ respectively, have the same best polysized representation $h$. But since, for every set $S \in B \Delta B'$, $r_{M_{B}}$ and $r_{M_{B'}}$, differ by a factor of $\Omega(n^{1/3}/\log^{2}(n))$, therefore, $h$ cannot approximate at least one of these two to a factor better $\Omega(n^{1/3}/\log^{2}(n))$. Since the value of $g$ and $g'$ at any point in the domain is off from that of $h$ by at most $1-\epsilon$, and hence it follows that $\alpha=\Omega(n^{1/3}/\log^{2}(n))$. 
\end{proof}